%% file: acyclicArXiv.tex
\documentclass{article}
\usepackage[english]{babel}
\usepackage{amsmath}
\usepackage{amssymb}
\usepackage{mathrsfs}
\usepackage{hyperref}
\usepackage{tikz}
\usetikzlibrary{arrows,automata,positioning, arrows.meta, calc}
\usepackage{delarray}
\usepackage{enumerate}
\usepackage{tabularx}
\usepackage{environ}
\usepackage{authblk}
\usepackage{stmaryrd}
\usepackage{amsthm}
\usepackage{authblk}

\tikzset{configuration/.style = {state, rectangle, minimum height=0.5cm},
}

\let\emptyset\varnothing

\let\implies\Rightarrow

\newcommand{\B}{\mathbb{B}}
\newcommand{\N}{\mathbb{N}}

\newcommand{\dom}{\mathrm{dom}}
\newcommand{\concat}{\cdot}

\makeatletter
\newcommand{\problemtitle}[1]{\gdef\@problemtitle{\hspace*{-3.5mm}\scalebox{.75}{$\blacktriangleright$} #1}}% Store problem title
\newcommand{\probleminput}[1]{\gdef\@probleminput{#1}}% Store problem input
\newcommand{\problemquestion}[1]{\gdef\@problemquestion{#1}}% Store problem question
\NewEnviron{cproblem}{
  \problemtitle{}\probleminput{}\problemquestion{}% Default input is empty
  \BODY% Parse input
  \par\addvspace{.5\baselineskip}
  \noindent
  \begin{tabularx}{\textwidth}{@{\hspace{\parindent}} l X c}
    \multicolumn{2}{@{\hspace{\parindent}}l}{\@problemtitle} \\% Title
    \textbf{Input:} & \@probleminput \\% Input
    \textbf{Question:} & \@problemquestion% Question
  \end{tabularx}
  \par\addvspace{.5\baselineskip}
}

\newtheorem{nclaim}{Claim}[section]
\newtheorem{example}{Example}[section]
\newtheorem{property}{Property}[section]
\newtheorem{corollary}{Corollary}[section]
\newtheorem{theorem}{Theorem}

\def\ech{.8}

\begin{document}
\title{%Characteristics of Synchronous Acyclic Modules
On the complexity of acyclic modules in automata networks}

\author[1]{K{\'e}vin Perrot}%\inst{1}
\author[2]{Pac{\^o}me Perrotin}%\inst{2}
\author[1]{Sylvain Sen{\'e}}
%\institute{
\affil[1]{Universit\'e publique}
%  \and
\affil[2]{Aix-Marseille Univ., Univ. de Toulon, CNRS, LIS, UMR 7020, Marseille, France}
%}

\date{}

\maketitle

\begin{abstract}
Modules were introduced as an extension of Boolean automata networks.  
They have inputs which are used in the computation said modules perform,
and can be used to wire modules
with each other.
In the present paper we extend this new formalism and study the specific case
of acyclic modules. 
These modules prove to be well described in their limit behavior by functions called 
output functions. 
We provide other results that offer an upper bound on the number of
attractors in an acyclic module when wired recursively into an automata network,
alongside a diversity of complexity results around the difficulty of deciding
the existence of cycles depending on the number of inputs and the size of said
cycle.
\end{abstract}

\section{Introduction}

Automata networks (ANs) are a generalisation of Cellular automata (CAs).
While classical CAs require a $n$-dimensional lattice with uniform local functions,
ANs can be built on any graph structure, and with any function at each vertex of the 
graph. They have been applied to the study of genetic regulation
networks~\cite{J-Kauffman1969,J-Thomas1973,J-Mendoza1998,J-Davidich2008,J-Demongeot2010} 
where the influence of different genes (inhibition, activation) are represented by
automata whose functions mirror together the global dynamics of the network.
This application in particular motivates the development of tools to understand,
predict and describe the dynamics of ANs in an efficient way. 
In the worst case, studying the dynamics of an AN (\emph{i.e.} analysing the behavior
of all possible configurations of the system) will always take an exponential
amount of time in the size of the network. 
Attempts using mainly combinatorics have been made to predict and count specific limit 
behavior of the system without enumerating the entire network's 
dynamics~\cite{J-Aracena2008,J-Demongeot2012,J-Aracena2017}. 
Other studies focused on understanding the dynamics of such complex systems by 
considering them as compositions of bricks simpler to 
analyse~\cite{J-Bernot2009,C-Siebert2009,C-Delaplace2012} and propose to study manners of 
controlling these bricks and/or systems~\cite{Biane2019,Pardo2019}. 
In line with such approaches and~\cite{T-Feder1990} the authors developed 
in~\cite{C-Perrot2018} the formalism of modules.
They are ANs with inputs, and operators called wirings that allow modules to be 
composed into larger modules, and eventually into ANs. 
In this paper we propose an exploration of a specific type of modules, namely acyclic 
modules, which do not include cycles in their interaction graph.
% providing definitions and results supporting that this formalism 
%allows for a an alternative understanding of the dynamics of ANs which does not
%require its explicit computation.\smallskip
The present paper also introduces output functions, which characterise the
behavior of an acyclic module as a function of the inputs of the network over
time. Output functions allow us to characterise the dynamics of a network
while forgetting its inner structure, illustrated by
Theorem~\ref{th-limit-from-output}, which shows that if two acyclic modules have
equivalent output functions, they also have isomorphic attractors.\smallskip

In Section~\ref{section-def} we propose definitions of ANs, modules and
wirings. 
Section~\ref{section-acyclic} presents definitions of acyclicity in modules and 
related concepts and results. 
Finally in Section~\ref{section-complexity} we explore complexity results
around acyclic modules and their inputs.
%The demonstrations of all 
%results are available in the appendix of this paper.

\paragraph{General notations.} We denote $\B$ the set of Booleans $\B = \{0, 
1\}$.
For $\Lambda$ an alphabet, we denote $\Lambda^n$ the set of vectors of size
$n$ with values in $\Lambda$. 
For $x \in \Lambda^n$, we might denote $x$ by $x_1 x_2 \ldots x_n$. 
For example, a vector $x \in \B^3$ defined such that $x_1 = 1$, $x_2 = 0$, $x_3 = 1$ 
can alternatively be denoted by $x = 101$.
For $S$ an ordered set of labels, $x \in \Lambda^S$,
$s$ in $S$, and $f$ a function 
which takes $x$ as an input, we might denote $f(x) = s$ as a 
simplification of $f(x) = x_s$.
For $G$ a digraph, we denote by $V(G)$ the set of its vertices and
by $A(G)$ the set of its arcs.
Let $G, G'$ be two digraphs, we denote $G \subseteq G'$ if and only if $G$ is an induced
subdigraph of $G'$, that is $V(G) \subseteq V(G')$ and $u, v \in V(G)$ implies
$(u, v) \in A(G) \Leftrightarrow (u, v) \in A(G')$.
For $f: A \to B$, and $C \subseteq A$, we denote $f|_C$ the function defined
over $f|_C: C \to B$ such that $f|_C(x) = f(x)$ for all $x \in C$.
%An alphabet is any set of labels, generally denoted $\Lambda$.
For $x \in \Lambda^S$, %a vector over the set $S$,
for any function $f : R \to S$ (for some set $R$),
we define $x \circ f$ as %the vector such that
$(x \circ f)_r = x_{f(r)}$, for
all $r \in R$. For $X = (x_1, x_2, \ldots, x_k)$ a sequence of $x_i \in \Lambda^S$,
%vectors,
we
define $X \circ f$ as the sequence $(x_1 \circ f, x_2 \circ f, \ldots, x_k \circ f)$.
In most of our examples, the alphabet $\Lambda$ will be $\B$ and the set $S$ finite, hence $x \in \Lambda^S$ will be considered as a Boolean vector (according to some order on $S$).

\section{Definitions}
\label{section-def}

\subsection{Automata networks}

ANs are composed by a set $S$ of automata. %, generally denoted $S$.
Each automaton in $S$, or node, is at any time in a state in $\Lambda$.
Gathering those isolated states into a vector of dimension $|S|$ provides us 
with a configuration of the network. 
More formally, a \emph{configuration} of $S$ over $\Lambda$ is a vector in 
$\Lambda^S$.
%\begin{definition}[Configuration]
%Let $S$ be a set of labels, and $\Lambda$ an alphabet. A \emph{configuration} of
%$S$ over $\Lambda$ is a vector in $\Lambda^S$.
%\end{definition}
The state of every automaton is bound to evolve as a function of the
configuration of the entire network. Each node has a unique function, called
a local function that is predefined and does not change over time. 
A \emph{local function} is thus a function $f$ defined over $f: \Lambda^S \to 
\Lambda$.
%\begin{definition}[Local function]
%For $S$ a set of labels, and $\Lambda$ an alphabet, a \emph{local function} is
%a function $f$ defined over $f: \Lambda^S \to \Lambda$.
%\end{definition}
An AN is described as a set which provides a local function to every node in the network. 
Formally, an \emph{automata network} $F$ is a set of local functions $f_s$ over 
$S$ and $\Lambda$ for every $s \in S$.
%\begin{definition}[Automata Network]
%Let $S$ be a set of labels, and $\Lambda$ an alphabet. An \emph{Automata Network}
%$F$ is a set of local functions $f_s$ over $S$ and $\Lambda$ for every $s \in S$.
%\end{definition}

\begin{example}
\label{example-F}
For $\Lambda = \B$, and $S = \{a, b, c\}$, let $F$ be a Boolean AN with local functions 
$f_a(x) = \neg a$, $f_b(x) = a \vee \neg c$, and $f_c(x) = \neg c \wedge \neg a$.
\end{example}

The configuration of an AN is updated using the local functions. 
The protocol by which the local functions are applied is called its update 
schedule.
Many different update schedules exist (actually, there are an infinite number of 
these), and it is well known %it has been shown 
that changing the update schedule of ANs can change the obtained 
dynamics~\cite{L-Robert1986,J-Goles2008,J-Aracena2013,J-Noual2018}.
The update schedule used in this paper is the parallel update schedule, in which every 
node udpates its value according to its local function at each time step. 
%in a synchronous way.
Thus, considering a configuration $x$ of an AN $F$, the \emph{update} $F(x)$ of 
$F$ over $x$ is the configuration such that for all $s \in S$, $F(x)_s = f_s(x)$, where 
$f_s$ is the local function assigned to $s$ in $F$.
%\begin{definition}[Automata Network Update]
%Let $x$ be a configuration of an Automata Network $F$. The update of $F$ over
%$x$, denoted $F(x)$, is defined as a configuration such that for all $s \in
%S$, $F(x)_s = f_s(x)$, where $f_s$ is the local function assigned to $s$ in $F$.
%\end{definition}

\begin{example}
Following the previous example, we can see that $F(000) = 111$, $F(010) = 111$
and that $F(111) = 010$.
\end{example}

ANs are usually represented by the influence that automata hold on each other.
As such the visual representation of an AN is a directed graph, called an interaction digraph, whose nodes are the automata of the network,
and arcs are the influences that link the different automata. 
Formally, $s$ \emph{influences} $s'$ if and only if there exist two 
configurations $x,x'$ such that $\forall r \in S, r \neq s \iff
x_r = x'_r$, and $F(x)_{s'} \neq F(x')_{s'}$.
%\begin{definition}[Influence]
%Let $F$ be an Automata Network, and $s, s' \in S$. We define that $s$
%\emph{influences} $s'$ if and only if there exist two configurations $x,x'$
%such that $\forall r \in S, r \neq s \Leftrightarrow
%x_r = x'_r$, and $F(x)_{s'} \neq F(x')_{s'}$.
%\end{definition}
From this, we define the
\emph{interaction digraph} of $F$ as the directed graph with nodes $S$ such that
$(s, s')$ is an arc of the digraph if and only if $s$ influences $s'$.
%\begin{definition} [Interaction Digraph]
%For $F$ an Automata Network over a set $S$, we define the
%\emph{interaction digraph} of $F$ as the directed graph with nodes $S$ such that
%$(s, s')$ is an arc of the digraph if and only if $s$ influences $s'$.
%\end{definition}
%\input{figures/automata}
For instance the interaction digraph of the network developed in 
Example~\ref{example-F} is depicted in Figure~\ref{fig-aut-dyn}.

To encapsulate the entire behavior of the network, one needs to enumerate all
the possible configurations the network, namely the elements of $\Lambda^S$, 
and describe the global update function upon this set.
This is often done via another graphical representation, which is another
digraph, called the dynamics of the network. Intuitively, this graph defines an arc from 
$x$ to $x'$ if and only if the update of the network over the configuration $x$ results 
in the configuration $x'$. Formally, the \emph{dynamics} of
$F$ can be represented as the digraph $G$ with vertex set $\Lambda^{S}$, such that $(x, x')$
is an arc in $G$ if and only if $F(x) = x'$.
%\begin{definition} [Automata Network Dynamics]
%For $F$ an Automata Network over a set $S$, we define the \emph{dynamics} of
%$F$ as the directed graph $G$ with vertex set $\Lambda^{S}$, such that $(x, x')$
%is an arc in $G$ if and only if $F(x) = x'$.
%\end{definition}
\input{figures/automata-dynamics}
The dynamics of the network developed in Example~\ref{example-F} is presented
in Figure~\ref{fig-aut-dyn}.

The dynamics of a network is a large object. A commonly studied part of this
object is called the attractors of the networks. An attractor is a sequence
of configurations which constitutes a cycle in the dynamics of the network.
Alternatively, the attractors of a network can be defined as the set of
non trivial strongly connected components of its dynamics. 
Formally, an \emph{attractor} of $F$ is a connected component of the subdigraph $G_L \subseteq G$, such 
that $x$ is a node in $G_L$ if and only if there exists $k \in \N \setminus\{0\}$ such 
that $F^k(x) = x$. Notice that, classically in the domain of ANs, An attractor of size 
one is called a \emph{fixed point}, whereas an attractor of size
greater than one is called \emph{a limit cycle}.
%\begin{definition} [Attractors]
%Let $F$ be an Automata Network and $G$ its dynamics. We define the
%\emph{attractors} of $F$ as the subdigraph $G_L \subseteq G$,
%such that $x$ is a node in $G_L$ if
%and only if there exists $k \in \N \setminus\{0\}$ such that $F^k(x) = x$.
%\end{definition}

\begin{example}
In our example, the attractors of $F$ are the configurations $010$ and $111$
since they verify $F^2(010) = 010$ and $F^2(111) = 111$. For any other
configuration, updating the network more than two times changes the state
of the network to $010$ or $111$. Alternatively, the configuration $010$ and
$111$ form the only non trivial strongly connected component of the 
dynamics of this network.
\end{example}

%An attractor of size one is called a fixed point, and an attractor of size
%greater than one is a limit cycle.

\subsection{Modules}

Informally, modules can be described as ANs with inputs. 
More formally, for a given module, we introduce a new set of labels, usually
denoted $I$, which contains the inputs of the module. 
By convention, inputs will be denoted with Greek letters. 
A local function of a module does not only depend on the states of the automata of the network, but also on the evaluations of the inputs.
Inputs are not automata, and do not have a state; but it is interesting to
suggest that inputs are added nodes of the network that do not admit local 
functions. %, which do not define local functions.
Formally, by considering $S$ and $I$ as sets of labels, and $\Lambda$ as an alphabet, 
a \emph{module} is a set which, for every $s \in S$, defines a local function
$f_s : \Lambda^{S \cup I} \to \Lambda$.

%\begin{definition}[Modules]
%Let $S$ and $I$ be sets of labels, and $\Lambda$ an alphabet.
%A \emph{module} is a set which, for every $s \in S$, defines a local function
%$f_s : \Lambda^{S \cup I} \to \Lambda$.
%\end{definition}

\begin{example}
\label{example-M}
For $\Lambda = \B$, $S = \{a, b, c\}$ and $I = \{\alpha, \beta, \gamma\}$
let $M$ be a module with local functions
$f_a(x, i) = \neg b \vee \alpha$,
$f_b(x, i) = a \vee \neg c \vee \beta \vee \neg \alpha$, and
$f_c(x, i) = \neg c \wedge \neg \gamma$.
\end{example}

The digraph representation of a module is similar to that of an AN;
the inputs are added for clarity as incident arrows to the nodes they
influence. 
For instance, the module of Example~\ref{example-M} is illustrated in
Figure~\ref{fig-acyclic}.
As well, updating a module over the parallel update schedule is similar to updating an
AN. The inputs are introduced with specific notations which are detailed below. 
Let $x$ and $i$ be configurations over $S$ and $I$ respectively. The update of a 
module $M$ over $x$ and $i$, denoted $M(x, i)$, is defined as a configuration over $S$ 
such that, for all $s \in S$, $M(x,i)_s = f_s(x,i)$, where $f_s$ is the local function 
assigned to $s$ in $M$.
%\begin{definition}[Module Update]
%Let $x$ and $i$ be configurations over label sets $S$ and $I$ respectively.
%The update of a module $M$ over $x$ and $i$, denoted $M(x, i)$, is defined as
%a configuration over $S$ such that, for all $s \in S$, $M(x,i)_s = f_s(x,i)$,
%where $f_s$ is the local function assigned to $s$ in $M$.
%\end{definition}

\begin{example}
Let us update the module $M$ over the node configuration $x = 011$ and the input
configuration $i = 000$. We compute $f_a(x, i) = \neg 1 \vee 0 = 0$, 
$f_b(x, i) = 0 \vee \neg 1 \vee 0 \vee \neg 0 = 1$ and 
$f_c(x, i) = \neg 1 \wedge \neg 0 = 0$, thus giving $M(011,000) = 010$.
\end{example}

Since it will be convenient to update a module over multiple iterations
at once, we will generally consider a sequence of input configurations of
the form $(i_1, i_2, \ldots, i_m)$. For $\alpha, \beta, \ldots$ the inputs
of the considered module, we will denote for convenience $\alpha_1, \beta_1, \ldots$
the evaluation of those inputs in the configuration $i_1$, and so on, denoting
$\alpha_k, \beta_k, \ldots$ the evaluation of the respective inputs in the
configuration $i_k$. We will denote by $M(x, (i_1, i_2, \ldots, i_m))$ the
execution of $m$ updates of the module $M$ starting with configuration $x$,
taking the input configuration $i_k$ at update number $k$. 
Formally, it is defined recursively as:
\[M(x, (i_1, i_2, \ldots i_m)) = M(M(x, i_1), (i_2, \ldots, i_m))
\text{, with } M(x,\emptyset)=x.\]

\subsection{Wirings}

Modules are a formalism of composition and decomposition of ANs. As such, we
define the process of composing modules together as wiring. 
Wirings exist in two forms. 
One is recursive, and proposes the rearrangement of a single module by connecting inputs 
of the module to itself. 
The second type of wiring is non-recursive, and defines the combination
of two modules into one, connecting inputs of one module to the nodes of
the other. 
When an input is connected, any function depending on the value
of that input relies on the state of the connected node instead.
Those two sorts of wirings were proven to be universal to compose any
network from elementary parts~\cite{C-Perrot2018}. 
Wiring operations are defined upon an object that specifies the operated connections, 
usually denoted $\omega$ which is a partial function defined from a subset of 
inputs of the second module to nodes of the first.

\input{figures/wirings}

%\begin{definition}[Recursive wiring]
\smallskip
\noindent{\bf Recursive wiring.} \emph{Let $M$ be a module with label sets $S$ and $I$ which, for every $s \in S$, 
defines the local function $f_s$. For $\omega : I \nrightarrow S$
a partial function, we define $\circlearrowright_\omega M$ the module which,
%for every $s \in S$ defines the local function $f'_s$ as the result of following
%substitution:
%\[f'_s = f_s[\alpha \mapsto \omega(\alpha) \mid \forall \alpha \in dom(\omega)].\]
for every $s \in S$, defines the local function $f'_s$ such that:
\begin{equation*}
	\forall x \in \Lambda^{S \cup I \setminus \dom(\omega)},\ f'_s(x) = f_s( x \circ \hat \omega ),
	\text{ with } \hat\omega(k) = \begin{array}\{{ll}.
  		\omega(k) & \text{if } k \in \dom(\omega)\\
		k & \text{if } k \in S \cup I \setminus \dom(\omega)
	\end{array}\text{.}
\end{equation*}}
%\end{definition}

\vspace{-0.4cm}

\begin{example}
For $\Lambda = \B$, $S = \{a, b, c\}$ and $I = \{\alpha, \beta, \gamma\}$
let $M$ be a module with local functions
$f_a(x, i) = \neg b \vee \alpha$,
$f_b(x, i) = a \vee \neg c \vee \beta \vee \neg \alpha$, and
$f_c(x, i) = \neg c \wedge \neg \gamma$.
Let us define a partial function $\omega : I \to S$ such that
$\dom(\omega) = \{\alpha, \gamma\}$, and $\omega(\alpha) = c$ and
$\omega(\gamma) = a$. The result of the recursive wiring
$\circlearrowright_\omega M$ is a module with label sets $S' = S$ and 
$I' = \{\beta\}$ with local functions
$f'_a(x, i) = \neg b \vee c$,
$f'_b(x, i) = a \vee \neg c \vee \beta \vee \neg c$, and
$f'_c(x, i) = \neg c \wedge \neg a$.
\end{example}

%\begin{definition}[Non-recursive wiring]
\noindent{\bf Non-recursive wiring.} \emph{Let $M$ and $M'$ be two modules with respective label sets $S,I$, and
$S',I'$. We denote $f_s$ and $f'_{s'}$ the local functions defined respectively
in $M$ and $M'$ for every $s \in S$ and $s' \in S'$. For $\omega : I' \nrightarrow S$ a partial function, we define
$M \rightarrowtail_\omega M'$ the module with label sets $S \cup S'$ and $I \cup I' 
\setminus dom(\omega)$
%which for every $s \in S \cup S'$ defines the local function $f''_s$ as the result of
%the following substitution:
%\begin{equation*}
%f''_s = \begin{array}\{{ll}.
%  f_s & \text{ if } s \in S\\[.5em]
%  f'_s[\alpha \mapsto \omega(\alpha) \mid \forall \alpha \in dom(\omega)] & \text{ if } s \in S'
%\end{array}\text{.}
%\end{equation*}
which, for every $s \in S \cup S'$, defines the local function $f''_s$ such that:
\begin{multline*}
	\forall x \in \Lambda^\mathscr{S},\ f''_s(x) = \begin{array}\{{ll}.
  	f_s(x|_{S \cup I}) & \text{ if } s \in S\\%[.5em]
 	f'_s(x \circ \hat\omega) & \text{ if } s \in S'
	\end{array}\\
	\text{ with }
	\hat\omega(k) = \begin{array}\{{ll}.
	\omega(k) & \text{if } k \in \dom(\omega)\\
	k & \text{if } k \in \mathscr{S}
	\end{array}\text{,}\\
    \text{for } \mathscr{S} = S \cup S' \cup I \cup I' \setminus \dom(\omega) \text{.}
\end{multline*}}

\vspace{-0.6cm}
%\end{definition}

\begin{example}
\label{example-omega}
For $\Lambda = \B$, $S = \{a, b, c\}$ and $I = \{\alpha, \beta, \gamma\}$,
let $M$ be a module with local functions
$f_a(x, i) = \neg b \vee \alpha$,
$f_b(x, i) = a \vee \neg c \vee \beta \vee \neg \alpha$, and
$f_c(x, i) = \neg c \wedge \neg \gamma$.
Let also be $S' = \{d, e\}$, $I' = \{\delta\}$ and $M'$ another module
with local functions $f'_d(x, i) = \neg d \vee e \vee \delta$ and
$f'_e(x,i) = \neg e \vee d$.
Let $\omega : I' \to S$ be the function such that $\omega(\delta) = b$.
The result of the non-recursive wiring $M \rightarrowtail_\omega M'$ is the
module with sets $S'' = \{a, b, c, d, e\}$ and $I'' = \{\alpha, \beta, \gamma\}$
with local functions
$f''_a(x, i) = \neg b \vee \alpha$,
$f''_b(x, i) = a \vee \neg c \vee \beta \vee \neg \alpha$,
$f''_c(x, i) = \neg c \wedge \neg \gamma$,
$f''_d(x, i) = \neg d \vee e \vee b$ and
$f''_e(x,i) = \neg e \vee d$. (See an illustration in Figure~\ref{fig-wiring}.)
\end{example}

\section{Acyclicity}
\label{section-acyclic}

\subsection{Acyclic automata networks}

%Acyclicity is a property of graphs. 
%In the context of ANs, acyclicity 
Acyclicity is a property of the interaction digraph of the considered AN; it means that 
no node of the network influences itself, neither by a direct loop 
nor through the action of any cycle that would include this node. 
Acyclic ANs have been one of the first families of ANs to be studied and 
characterised~\cite{L-Robert1986}. Their dynamical behavior is trivial:
there is only one fixed point, which attracts every other configuration. This
is true under the parallel update schedule as well as any other schedule which
would eventually update every node a minimum amount of time for the stabilisation
to propagate. This early result led to the simple conclusion that cycles are essential 
to the complexity of their dynamics.

%\vspace{-0.2cm}

\subsection{Acyclic Modules}

%We now propose the study of the same property of interaction digraphs over the
%formalism of modules.

\noindent{\bf Acyclicity.} \emph{A module $M$ is \emph{acyclic} if its interaction digraph is acyclic.}
%\begin{definition}[Acyclic module]
%A module $M$ is \emph{acyclic} if its interaction digraph is acyclic.
%\end{definition}

\begin{example}
\label{example-acyM}
For $\Lambda = \B$, $S = \{a, b, c\}$ and $I = \{\alpha, \beta, \gamma\}$
let $M$ be a module with local functions
$f_a(x, i) = \alpha$,
$f_b(x, i) = a \vee \beta \vee \neg \alpha$, and
$f_c(x, i) = \neg b \wedge a \wedge \neg \gamma$.
$M$ is acyclic. (See an illustration in Figure~\ref{fig-acyclic}.)
\end{example}

The dynamics of this family of objects is simple enough to be studied, and 
complex enough to provide insights into the general dynamics of ANs.
It is indeed clear that every AN can be decomposed into a
recursively wired acyclic module. This can be done by taking a feedback arc
set of the interaction digraph of the network, and producing a module that
replaces every arc in the set by an input.

As an acyclic module has no loop or cycle in its influences, it can support no
long lasting memory used for computation. As such the behavior of any node in
the network can be understood as a function of only the evaluation of the inputs
in its last iterations.
This function is called an output function and how much it must look in the past
to make its prediction is called the delay of the function.
\input{figures/acyclic-module}

%\smallskip
%\noindent{\bf Output function.} \emph{For $k \in \N$, An \emph{output function} $O$ of delay $m$ is a function defined over
%$\Lambda^{m \times k} \to \Lambda$.}\smallskip
%\begin{definition}[Output function] 
%For $k \in \N$, An \emph{output function} $O$ of delay $m$ is a function defined over
%$\Lambda^{m \times k} \to \Lambda$.
%\end{definition}

For $M$ a module with $k$ inputs, an output function $O$ with delay $m$ is a
function defined over a sequence of inputs $(i_1, i_2, \ldots, i_m)$. 
%An output function is said to have minimal delay if and only if parameter $i_m$
%has an
%influence on the computation of $O$. For any $O$ with non minimal delay,
%we can provide a minimal equivalent by taking the function $O'$ with largest
%delay $m'< m$ such that
%parameter $i_{m'}$ has an influence on the computation of $O'$, and
%$O|_{\Lambda^{m' \times k}} = O'$. For the rest of this paper, only output
%function with minimal delay will be used.
Each node of a network defines its own output function, similarly to how it
defines a local function. The output function of a node always has minimal delay
and will depend on the
output functions defined by the nodes which influence it. In other terms, if
node $a$ influences node $b$, then whatever output function which predicts the
value of $a$ based only on inputs will be useful to predict the evaluation of $b$
one iteration later. As such $b$ does not directly depend on the output function
of $a$, but on the output function of $a$ with incremented delay.

%\begin{definition}[Output function incrementation]
%Let $O$ be an output function of delay $m$. The \emph{incrementation} of $O$
%is the output function of delay $m+1$ denoted $O^{+1}$ which verifies:
%\[O^{+1}(i_1, i_2, \ldots, i_{m+1}) = O(i_2, i_3, \ldots, i_{m+1})\]
%for any sequence of input configurations $\{i_1, i_2, \ldots, i_{m+1}\}$.
%\end{definition}

Output functions are sufficient to describe the behavior of the entire module
from the inputs after a given amount of time. This fact is illustrated by
the Property~\ref{property-outputfunctions} below.

\noindent{\bf Node output.} \emph{Let $M$ be an acyclic module. For every $s \in S$, we define the output function
of $s$, denoted $O_s$, as the output function with minimal delay $m$ such that
for any sequence of inputs $J = (i_1, i_2, \ldots, i_m)$ and any
configuration $x$, $M(x, J)_s = O_s(J)$.}

\begin{example}
For $\Lambda = \B$, $S = \{a, b, c\}$ and $I = \{\alpha, \beta, \gamma\}$
let $M$ be a module with local functions
$f_a(x, i) = \alpha$,
$f_b(x, i) = a \vee \beta \vee \neg \alpha$, and
$f_c(x, i) = \neg b \wedge a \wedge \neg \gamma$.
The module $M$ verifies the following output functions :
$O_a = \alpha_1$, which has delay $1$,
$O_b = \alpha_2 \vee \beta_1 \vee \neg \alpha_1$, which has delay $2$, and
$O_c = \neg \alpha_3 \wedge \neg \beta_2 \wedge \alpha_2 \wedge \alpha_2 \wedge
\neg \gamma_1$, which has delay $3$.
\end{example}

%\begin{definition}[Node output]
%Let $M$ be an acyclic module. For every $s \in S$, we define the output function
%of $s$, denoted $O_s$, as the output function with minimal delay $m$ such that
%for any sequence of inputs $J = \{i_1, i_2, \ldots, i_m\}$ and any
%configuration $x$, $M(x, J)_s = O_s(J)$.
%\end{definition}

%\begin{property}
%\label{property-outputfunctions}
%Let $M$ be an acyclic module with $k$ inputs. For every $s \in S$, there exists
%an output function $O_s$ with delay $m$ which
%for any sequence of inputs $J = \{i_1, i_2, \ldots, i_m\}$ and any
%initial configuration $x$ verifies $M(x, J)_s = O_s(J)$.
%\end{property}
%\begin{proof}
%We can construct the output function of every node using the following
%algorithm : Initially, set $A = \emptyset$ and $B = S$. Then start a loop
%in which you put in $A$ every node in $B$ which function only depend on inputs
%and on the state of nodes in $A$. Create the output function of the nodes you
%move to $A$ by replacing the evaluation of any node in $A$ by the incrementation
%of the output function of that node.
%\end{proof}

\begin{property}
\label{property-outputfunctions}
Let $M$ be an acyclic module. For every $s \in S$, $s$ has one and only one
output function $O_s$.
\end{property}

\begin{proof}
  We first claim that there always exist an output function for any node:
  \begin{nclaim}
    \label{property-outputfunctions-fact1}
	Let $M$ be an acyclic module with $k$ inputs. For every $s \in S$, there exists
	an output function $O_s$ with delay $m$ which
	for any sequence of inputs $J = \{i_1, i_2, \ldots, i_m\}$ and any
	initial configuration $x$ verifies $M(x, J)_s = O_s(J)$.
  \end{nclaim}

  Let us define the incrementation of an output function.

  \smallskip
  \noindent{\bf Output function incrementation.} \emph{Let $O$ be an output function of 
  delay $m$. The \emph{incrementation} of $O$
  is the output function of delay $m+1$ denoted $O^{+1}$ such that $O^{+1}(i_1, i_2, 
  \ldots, i_{m+1}) = O(i_2, i_3, \ldots, i_{m+1})$ for any sequence of input configurations 
  $(i_1, i_2, \ldots, i_{m+1})$.}\smallskip

  To prove~\ref{property-outputfunctions-fact1}, see that
  $M$ is an acyclic module, therefore there exists a node $s \in S$ such that
  $s$ is not influenced by any node in $S$ (but possibly by some inputs).
  As a consequence, there exists an output function $O_{s}$ which simply equals
  $f_{s}$, and has a delay of $j_s=1$.  Now for the induction, consider the
  module $M$ without some set of nodes $S' \subset S$ such that for each node $s'
  \in S'$ we have already defined an output function $O_{s'}$ with delay
  $j_{s'}$.  Since it is still acyclic there exists a node $s \in S \setminus
  S'$ such that $s$ is not influenced by any node in $S \setminus S'$ (but
  possibly by some inputs and some nodes in $S'$).  As a consequence, there
  exists an input function $O_{s}$ which computes the local function $f_{s}$,
  replacing the evaluation of any node $s' \in S'$ by the incrementation of the
  output function $O_{s'}$, and has a delay of $j_s= 1+\max\{j_{s'} \mid s' \in
  S' \}$.

  We now make the following claim:
  \begin{nclaim}
    \label{property-outputfunctions-fact2}
	Let $M$ be an acyclic module. Let $s \in S$, and
	$O_s$ and $O'_s$ be two output functions with respective delays $m$ and $m'$
	such that for any two sequences of inputs $J, J'$ of size $m$ and $m'$
	respectively and any initial configuration $x$, $M(x, J)_s = O_s(J)$ and
	$M(x, J')_s = O'_s(J)$. If $m = m'$, then $O_s = O'_s$.
  \end{nclaim}

  To see this is true, suppose $m = m'$. This implies that $J$ and $J'$ are of
  the same size. For any $J$ such that $J = J'$, we verify
  $O_s(J) = M(x, J)_s = M(x, J')_s = O'_s(J') = O'_s(J)$. Therefore $O_s = O'_s$.

  We conclude by stating that any two different minimal output function for $s$
  would provide a contradiction with claim~\ref{property-outputfunctions-fact2}.
  
\end{proof}

Property~\ref{property-outputfunctions} can be further refined to propose the
following result, which
states that two networks have the same attractors if and only if the modules they
can be decomposed into have the same number of inputs and the same output
functions on the nodes on which those inputs are wired. As such, modules can
be considered as black boxes which are to be considered equivalent in their limit
behavior, as long as they share the same output functions, according to some
bijection between their inputs.

\begin{theorem}
\label{th-limit-from-output}
	Let $M$ and $M'$ be two acyclic modules, with $T$ and $T'$ subsets of their nodes
	such that $|T| = |T'|$. If there exists $g$ a bijection from $I$ to $I'$ and
	$h$ a bijection from $T$ to $T'$ such that for
	every $s \in T$, $O_s$ and $O'_{h(s)}$ have same delay, and for every input
	sequence $J$ with length the delay of $O_s$,
	\[O_s(J) = O'_{h(s)}(J \circ g^{-1})\]
	then for any function $\omega : I \to T$, the networks
	$\circlearrowright_{\omega} M$ and $\circlearrowright_{h \circ \omega \circ g^{-1}} M'$ have
        isomorphic attractors (up to the renaming of automata given by $h$).
\end{theorem}

\begin{proof}
  First remark that $\omega$ has domain $I$ hence it wires all inputs of $M$
  and therefore $\circlearrowright_{\omega} M$ is an automata network with
  a dynamics and attractors. Furthermore $g$ is a bijection from $I$ to $I'$
  hence the same applies to $\circlearrowright_{h \circ \omega \circ g^{-1}} M'$.
	Let us denote
	$F = \circlearrowright_{\omega} M$ and
	$F' = \circlearrowright_{h \circ \omega \circ g^{-1}} M'$
	for simplicity, with $G$ and $G'$ the dynamics restricted to their
	respective attractors. We want to show that $G$ and $G'$ are isomorphic.

	For $x \in \Lambda^S$, we define the input sequence of length $k$ generated
	by $x$, denoted $\hat{J}(x)^k$, as the sequence which verifies
        \[\hat{J}(x)^k_\ell = F^{\ell-1}(x)|_T \circ \omega \text{, for } 1 \leq \ell \leq k.\]
	Intuitively, the sequence $\hat{J}(x)^k$ records the evaluation of
	the network's outputs on $T$, over $k$ consecutive updates,
	starting with configuration $x$.

	\begin{nclaim}
    \label{th-limit-fact1}
	Let $k$ be such that $\forall s \in S$ with $d_s \leq k$, for $d_s$ the delay of the
	output function $O_s$ in $M$. For $J$ an input sequence of length $k$, the
	evaluation of $M(x, J)$ is always the same, regardless of the starting
	configuration $x \in \Lambda^S$.
	\end{nclaim}

	To see that this is true, apply Property~\ref{property-outputfunctions}
	and consider that $M(x, J)_s = O_s(J)$, for every $s \in S$.
	This computation is properly defined
	as per the definition of the length of $J$. It follows that the computation
	of $M(x, J)$ only depends on $J$. Based on this fact, we will denote
	$M(x, J) = M(J)$ in the rest of this demonstration when applicable,
	that is, when no output function of $M$ has a delay greater than $k$.

	\begin{nclaim}
    \label{th-limit-fact2}
	Let $J$ be an input sequence  of length $k$ such that the configuration
	$M(J)$ is defined.
	$\hat{J}( M(J) )^k = J \Rightarrow M(J) \in V(G)$.
	\end{nclaim}

	This Claim states that if the configuration $M(J)$, which is obtained by
	updating any configuration $x$ in $M$ with the input sequence $J$, generates
	the input sequence $J$, then $M(J)$ is a configuration which belongs to an
	attractor of $F$.

	Let us denote $x = M(J)$. By hypothesis, $\hat J(x)^k = J$. It follows that:
	\begin{align*}
          F^k(x) =& F( F^{k-1}(x) ) = M(F^{k-1}(x), F^{k-1}(x)|_T \circ \omega) \\
          =& M( M( \ldots M( x, F^0(x)|_T \circ \omega) \ldots
		, F^{k-2}(x)|_T \circ \omega), F^{k-1}(x)|_T \circ \omega)\\
          =& M( M( \ldots M( x, \hat J(x)^k_{1} ) \ldots
		, \hat J(x)^k_{k-1}), \hat J(x)^k_k) \\
          =& M(x, \hat J(x)^k) = M(x,J) = M(J) = x
	\end{align*}
	which implies that $F^k(x) = x$ and $x$ is part of an attractor which length
	divides $k$, hence the Claim holds.

	\begin{nclaim}
    \label{th-limit-fact3}
	Let $x \in V(G)$. There exists $x' \in V(G')$ such that
	$\hat J(x)^k \circ g^{-1} = \hat J(x')^k$, for every $k \in \N$.
	\end{nclaim}

	This Claim implies that, for any configuration $x$ in an attractor of $F$,
	there exists a configuration $x'$ in an attractor of $F'$ which generates
	an input sequence $\hat J(x')^k$ equivalent to the input sequence $\hat J(x)^k$
	up to the bijection $g$, and that holds for any length $k$.

	To prove it, consider $x \in V(G)$ and let us take $k$ 
	greater than the the delay of any output function in $M$ and $M'$;
	and such that $F^k(x) = x$.
	We consider the input sequences $\hat J(x)^k$ and $\hat J(x)^k \circ g^{-1}$.
	Claim~\ref{th-limit-fact1} implies that $M'( \hat J(x)^k \circ g^{-1})$ is
	a well defined configuration over $M'$, which we shall denote $x'$.
	Let us prove that $\hat J(x)^k \circ g^{-1} = \hat J(x')^k$.
	By definition we know that
	\[\hat J(x)^k_1 \circ g^{-1} = F^0(x)|_T \circ \omega \circ g^{-1} = x|_T \circ \omega \circ g^{-1}\]
	while
	\begin{align*}
	  \hat J(x')^k_1 = F'^0(x')|_{T'} \circ h \circ \omega \circ g^{-1}
          =& x'|_{T'} \circ h \circ \omega \circ g^{-1}\\
	  =& M'( \hat J(x)^k \circ g^{-1})|_{T'} \circ h \circ \omega \circ g^{-1}.
	\end{align*}
	Let us note that for any $s' \in T'$,
	$M'( \hat J(x)^k \circ g^{-1})_{s'} = O'_{s'}( \hat J(x)^k \circ g^{-1} )$
	which equals $O_{h^{-1}(s')} ( \hat J(x)^k )$ by the hypothesis of the Theorem.
	It follows that
        \[M'( \hat J(x)^k \circ g^{-1})|_{T'} \circ h = M( \hat J(x)^k )|_T \circ h^{-1} \circ h = x|_T\]
	and this implies that
	\[ \hat J(x')^k_1 = x|_T \circ \omega \circ g^{-1} = \hat J(x)^k_1 \circ g^{-1} \]
	therefore $\hat J(x)^k_1 \circ g^{-1} = \hat J(x')^k_1$.

	This marks the first step of the induction to prove $\hat J(x)^k \circ g^{-1} = \hat J(x')^k$.
	Let us state the induction hypothesis that
        \[\hat J(x)^k_{[1,\ell]} \circ g^{-1} = \hat J(x')^k_{[1,\ell]} \text{, for } \ell < k.\]
	We now prove that it implies
	$\hat J(x)^k_{[1,\ell+1]} \circ g^{-1} = \hat J(x')^k_{[1,\ell+1]}$. To prove it,
	we only need to prove $\hat J(x)^k_{\ell+1} \circ g^{-1} = \hat J(x')^k_{\ell+1}$.
	Let $\concat$ denote the concatenation of two sequences. We know that 
	\begin{align*}
          \hat J(x)^k_{\ell+1} \circ g^{-1} =& F^\ell(x)|_T \circ \omega \circ g^{-1}\\
          =& M(x, \hat J(x)^k_{[1,\ell]})|_T \circ \omega \circ g^{-1}\\
          =& M(M(\hat J(x)^k), \hat J(x)^k_{[1,\ell]})|_T \circ \omega \circ g^{-1}\\
          =& M( \hat J(x)^k \concat \hat J(x)^k_{[1,\ell]})|_T \circ \omega \circ g^{-1}
	\end{align*}
	and
	\begin{align*}
          \hat J(x')^k_{\ell+1} =& F'^\ell(x')|_{T'} \circ h \circ \omega \circ g^{-1}\\
          =& M'(x', \hat J(x')^k_{[1,\ell]})|_{T'} \circ h \circ \omega \circ g^{-1}\\
	  =& M'(x', \hat J(x)^k_{[1,\ell]} \circ g^{-1})|_{T'} \circ h \circ \omega \circ g^{-1}\\
	  =& M'(M'( \hat J(x)^k \circ g^{-1}), \hat J(x)^k_{[1,\ell]} \circ g^{-1})|_{T'} \circ h \circ \omega \circ g^{-1}\\
	  =& M'( (\hat J(x)^k \circ g^{-1}) \concat (\hat J(x)^k_{[1,\ell]} \circ g^{-1}))|_{T'} \circ h \circ \omega \circ g^{-1}\\
          =& M'( (\hat J(x)^k \concat \hat J(x)^k_{[1,\ell]}) \circ g^{-1})|_{T'} \circ h \circ \omega \circ g^{-1}.
	\end{align*}
        As the sequence $(\hat J(x)^k \concat \hat J(x)^k_{[1,\ell]}) \circ g^{-1}$ is at least
	of length $k$, we can use it to compute the result of output functions.
        From the hypothesis of the Theorem
        it follows that for every $s' \in T'$,
	\begin{align*}
          M'( (\hat J(x)^k \concat \hat J(x)^k_{[1,\ell]}) \circ g^{-1})_{s'}
          =& O'_{s'} ( (\hat J(x)^k \concat \hat J(x)^k_{[1,\ell]}) \circ g^{-1} )\\
	  =& O_{h^{-1}(s')} ( \hat J(x)^k \concat \hat J(x)^k_{[1,\ell]} )\\
          =& M( \hat J(x)^k \concat \hat J(x)^k_{[1,\ell]} )_{h^{-1}(s')}
	\end{align*}
	which, using again the hypothesis of the Theorem to relate $M$ and $M'$, implies
	\begin{align*}
          \hat J(x')^k_{\ell+1} =&
          M'( (\hat J(x)^k \concat \hat J(x)^k_{[1,\ell]}) \circ g^{-1})|_{T'} \circ h \circ \omega \circ g^{-1}\\
          =& M( \hat J(x)^k \concat \hat J(x)^k_{[1,\ell]} )_T \circ h^{-1} \circ h \circ \omega \circ g^{-1}\\
          =& M( \hat J(x)^k \concat \hat J(x)^k_{[1,\ell]} )_T \circ \omega \circ g^{-1}\\
          =& \hat J(x)^k_{\ell+1} \circ g^{-1}
	\end{align*}
	and concludes the induction, therefore
	$\hat J(x)^k \circ g^{-1} = \hat J(x')^k$. It follows that
	$\hat J(M'(\hat J(x)^k \circ g^{-1}))^k = \hat J(x')^k = \hat J(x)^k \circ g^{-1}$,
	which implies by Claim~\ref{th-limit-fact2} that $x' \in V(G')$, and that
	$x'$ is in an attractor which size divides $k$, just like $x$.
        This concludes
	our proof of Claim~\ref{th-limit-fact3} for $k$ big enough,
        but remark that as a consequence it holds for any $k \in \N$.

	Observe a symmetric sequence of arguments to prove that for every $x' \in V(G')$,
	there exists $x \in V(G)$ such that $\hat J(x')^k \circ g = \hat J(x)^k.$
	It follows that for any $x \in V(G)$, there exists a unique $x' \in V(G')$
	such that the above relation holds.
	This is true since if there existed $x', x'' \in V(G')$ such that
	$\hat J(x)^k \circ g^{-1} = \hat J(x')^k$, $\hat J(x)^k \circ g^{-1} = \hat J(x'')^k$,
	and $x' \neq x''$,
	then it would follow that
        $\hat J(x')^k = \hat J(x)^k \circ g^{-1} = \hat J(x'')^k \circ g \circ g^{-1} = \hat J(x'')^k$.
	Since $x', x'' \in V(G')$, for a large enough $k$ multiple
	of the sizes of the attractors containing $x'$ and $x''$,
        we would have
	$x' = M'(\hat J(x')^k) = M'(\hat J(x'')^k) = x''$, a contradiction.

	Let us therefore denote $\hat h : V(G) \to V(G')$ the bijection which to any
	$x \in V(G)$ associates $x' \in V(G')$ such that $\hat J(x)^k \circ g^{-1} = \hat J(x')^k$.
	This implies that $\hat h(x) = M'(\hat J(x)^k \circ g^{-1})$, for $k$ larger
	than the delay of any output function in $M$ and $M'$, and multiple of the
	size of the attractors which contain $x$ and $\hat h(x)$.
	Let us prove that $\hat h$ is an isomorphism from $G$ to $G'$. We need to
	prove that, for any $x \in V(G)$, $\hat h(F(x)) = F'(\hat h(x))$.

	Let $x \in V(G)$ and $k$ a multiple of the length of the attractor $x$ is
	part of, such that $k$ is greater than any delay of any output function
	in both $M$ and $M'$.
	It follows that
	\begin{align*}
          \hat h(F(x)) =&
	  M'( \hat J(F(x))^k \circ g^{-1})\\
          =& M'( ( F^0(F(x))|_T \circ \omega,
		  F^1(F(x))|_T \circ \omega, \ldots,
		  F^{k-1}(F(x))|_T \circ \omega ) \circ g^{-1})\\
          =& M'( ( F^1(x)|_T \circ \omega,
		  F^2(x)|_T \circ \omega, \ldots,
		  F^k(x)|_T \circ \omega ) \circ g^{-1})\\
          =& M'( ( F^1(x)|_T \circ \omega \circ g^{-1},
		  F^2(x)|_T \circ \omega \circ g^{-1}, \ldots,
		  F^k(x)|_T \circ \omega \circ g^{-1})).%\\
		%= M'( ( F'(x)^1|_{T'} \circ h \circ \omega \circ g^{-1},
		%		 F'(x)^2|_{T'} \circ h \circ \omega \circ g^{-1}, \ldots,
		%		 F'(x)^k|_{T'} \circ h \circ \omega \circ g^{-1}))\\
	\end{align*}
	Let us consider an individual element of the above sequence,
	$F^\ell(x)|_T \circ \omega \circ g^{-1}$. It follows that for every $s \in S$,
	\begin{align*}
          F^\ell(x)_s =& M(x, \hat J(x)^\ell)_s\\
          =& M(M(\hat J(x)^k), \hat J(x)^\ell)_s\\
	  =& M(\hat J(x)^k \concat \hat J(x)^\ell)_s\\
	  =& O_s(\hat J(x)^k \concat \hat J(x)^\ell)\\
          =& O'_{h(s)}( (\hat J(x)^k \concat \hat J(x)^\ell) \circ g^{-1})\\
          =& M'( (\hat J(x)^k \concat \hat J(x)^\ell) \circ g^{-1})_{h(s)}\\
	  =& M'( M'( \hat J(x)^k \circ g^{-1}), \hat J(x)^\ell \circ g^{-1})_{h(s)}\\
	  =& M'( \hat h(x), \hat J(x)^\ell \circ g^{-1})_{h(s)}\\
	  =& F'^\ell( \hat h(x) )_{h(s)}
	\end{align*}
	which implies that $F^\ell(x)|_T \circ \omega \circ g^{-1} =
	F'^\ell( \hat h(x) )|_{T'} \circ h \circ \omega \circ g^{-1}$. This, applied
	to the previous development, gives
	\begin{align*}
          \hat h(F(x)) =&
          M'( ( F'^1( \hat h(x) )|_{T'} \circ h \circ \omega \circ g^{-1},
		F'^2( \hat h(x) )|_{T'} \circ h \circ \omega \circ g^{-1},
                \ldots\\
                &\hspace*{5cm} \ldots,
                F'^k( \hat h(x) )|_{T'} \circ h \circ \omega \circ g^{-1}))\\
          =& M'( \hat J( F'( \hat h(x) ) ) )\\
          =& F'( \hat h(x) )
	\end{align*}
	and concludes the proof of the Theorem.
        
\end{proof}

Output functions are a characterisation of the behavior of acyclic modules which
is enough to understand their limit dynamics under parallel schedule.
This characterisation behaves in expected ways under non-recursive wirings.
Taking two acyclic modules and wiring them non-recursively makes a module whose
output functions are deducible from the output functions of the initial
acyclic module.
We now state a result which provides an upper bound on the
number of attractors of each size of an AN, which is wired from
a module with $k$ inputs.

\begin{theorem}
\label{th-bound}
Taking an acyclic module with $k$ inputs and wiring all inputs recursively gives
an AN. Let us denote $a(k, c)$ the number of attractors of size $c$ of its dynamics.
We state $a(k, c) \leq A(k, c)$, with:
  \[A(k, 1) = |\Lambda|^k
  \qquad \text{ and } \qquad
  A(k, c) = |\Lambda|^{kc} - \sum_{c' < c,\ c' | c} A(k, c').\]
\end{theorem}

\begin{proof}
	Let us consider an acyclic module $M$ with $k$ inputs. Consider a wiring $\omega$
	over $M$ such that $\dom(\omega) = I$, for $I$ the set of $k$ inputs of $M$.
	Finally consider the dynamics of the Automata Network
	$F = \circlearrowright_\omega M$. Let us denote $\omega(I)$ and call output set
	the set defined $\omega(I) = \{\omega(\alpha) \mid \alpha \in I\}$. We remark
	the following fact :

	\begin{equation}
		\label{boundproof-Isize}
		|\omega(I)| \leq |I| = k
	\end{equation}

	Let us consider an attractor $X = \{x_1, x_2, \ldots, x_c\}$ over $F$.
	By definition $F(x_i) = x_{i+1}$ for $i < c$ and $F(x_c) = x_1$.
	For $R \subseteq S$, and $x$ a vector over $S$ with values in $\Lambda$, we
	define $x|_R$ the projection of this vector over $R$. By extension, $X|_R$
	denotes the projection of the attractor $X$.
	Provided another such attractor $X'$ of same size, we make the following claim.
	\begin{nclaim}
		\label{boundproof-hypo}
		$X|_{\omega(I)} = X'|_{\omega(I)} \implies X = X'$.
	\end{nclaim}
	To see this is true, let us assume that $X|_{\omega(I)} = X'|_{\omega(I)}$.
	Since $M$ is acyclic by definition, we know that there exists a non empty set
	of nodes $S_1 \subseteq S$ such that every $s \in S_1$ is only influenced by
	inputs and not by any other node. This means that assuming
	$X|_{\omega(I)} = X'|_{\omega(I)}$
	implies $X|_{\omega(I) \cup S_1} = X'|_{\omega(I) \cup S_1}$. Now consider that,
	after the same acyclicity hypothesis, there exists a non-empty set
	$S_2 \subseteq S$ of nodes which are only influenced by inputs, and nodes in
	$S_1$, which implies $X|_{\omega(I) \cup S_1 \cup S_2} =
        X'|_{\omega(I) \cup S_1 \cup S_2}$. The claim follows
	by induction.

	This claim allows us to prove that there can only be as
	many attractors of size $c$ in $F$ as there is distinct $X|_{\omega(I)}$. This
	provides us with a weaker form of the result:
	\begin{equation}
		\label{boundproof-midstep}
		a(k, c) \leq |\Lambda|^{kc}
	\end{equation}
	Let $X$ be one of the $|\Lambda|^{kc}$ possible sequence of $c$ configurations.
	Let us assume that $F(x_i) = x_{i+1}$ for $i < c$ and $F(x_c) = x_1$.
	By definition, if there exists
	$i, j$ such that $i \neq j$ and $x_i = x_j$, the sequence $X$ will be periodic.
	This implies the existence of a smaller sequence $X'$ such that $X = X'^q$ for
	$q \in \N$. In other words, for every possible proper attractor $X'$ such that
	the size of $X'$ divides $c$, there exists a sequence $X = X'^{\frac{c}{|X'|}}$
	which is not an attractor of $F$ by definition. Using this fact
	and \ref{boundproof-midstep}, we conclude that $a(c,k)$ is
	not greater than $|\Lambda|^{kc} - \sum_{c' < c, c' | c} A(k, c')$.
        
\end{proof}

The smallest $k$ which can be provided for any AN is equal to the
minimum feedback arc set of the network. As such this result is very similar
to a previous result of~\cite{J-Aracena2008,J-Aracena2017}, which states an upper bound 
on the total number of attractors depending on the size of a positive feedback arc set.
Though the global bound with a positive feedback arc set would be stronger, the
present result is different as it operates on parallel update schedule and provides
different bounds on different sizes of attractors, where the previous result
offered a bound on the total count of attractors under asynchronous update
schedule.

\subsection{One-to-one modules}

A module with only one input has the particularity of being recursively wired
in a linear amount of possible ways. That is, the only degree of freedom in the wiring
is the choice of the node which will serve as output. Let us consider a module with only 
one input, and let us consider $e \in S$ as the designated output node of the module. In 
this context we will denote $\circlearrowright M$ as the AN obtained by wiring the input 
of the module to its designated output. Furthermore, the output function $O_e$ will 
sometimes be denoted $O$, as the designated output function of the module. Such an acyclic module 
with only one input and a designated output is called a one-to-one module.

\begin{theorem}
\label{th-one-to-one}
Let $M$ be a one-to-one module. The one-to-one module $M_{min}$ with a minimum
number of nodes and which defines the same output function as $M$ is of size $d$,
for $d$ the delay of the output function of $M$. 
\end{theorem}

\input{figures/one-to-one}

\begin{proof}
	First we can prove that we cannot construct a module with a size smaller than 
        the delay of its output function. This is easily shown as there need to be
        a line of at least $d$ in size
        in the network's interaction digraph.

	To prove that such a minimal network always exists, simply construct it by using
	$d-1$ nodes as a line which offers the input's value delays from $2$ to $d$.
	The last node computes the desired output function and takes the values from
	the input directly for a delay of $1$, or from the rest of the network for a
	delay from $2$ to $d$.
        
\end{proof}

An example of the application of Theorem~\ref{th-one-to-one} is illustrated in 
Figure~\ref{fig-one-to-one}. This construction is polynomial in time, and bears strong 
resemblances with the objects known as Feedback Shift Registers~\cite{B-Elspas1959}. 
%A good introduction to Feedback Shift Registers can be found in~\cite{B-Elspas1959}.

\section{Complexity Results}
\label{section-complexity}

%In the following section, we will denote "the recursive wiring of a module" as
%an Automata Network which is obtained by recursively wiring every input of said
%module to itself.

This section presents complexity results that have been obtained around
output functions, and the difficulty of the analysis of the dynamics of
acyclic modules after being recursively wired into a complete network.
Remark that these questions have been widely addressed in the context of
threshold Boolean ANs~\cite{J-Alon1985,C-Orponen1989,C-Orponen1992}. Such
a wiring will sometimes be denoted as a complete recursive wiring of the module.
A module is encoded into the input of a decision problem as the list of its
local functions written in propositional logic. As such the computation
of the output functions of an acyclic module is comparable to the computation
of a circuit.
%We first examine the complexity of the computation of an output function in
%an acyclic module. This process is equivalent to the computation of a
%circuit.

%\begin{cproblem}
%  \problemtitle{\textsc{Acyclic Module Output Problem}}
%  \probleminput{An acyclic module $M$, a node $s \in S$ and an output function $O$.}
%  \problemquestion{Does $O_s = O$ in $M$?}
%\end{cproblem}

%\begin{theorem}
%The Acyclic Module Output Problem is P-complete.
%\end{theorem}
%\begin{sproof}
%We reduce from the Circuit Value Problem, as the computation of an output
%function in a Boolean module without any input results in a constant
%output function which is equivalent to the evaluation of a circuit.
%
%\end{sproof}

Let us provide a few decision problems on the
dynamics of a network obtained from a recursively wired acyclic module.

\begin{cproblem}
  \problemtitle{\textsc{Acyclic Module Attractor Problem}}
  \probleminput{An acyclic module $M$ with $k$ inputs and $n$ nodes, a 
	function $\omega$ which defines a complete recursive wiring over $M$,
	and a number $c$ encoded in unary.}
  \problemquestion{Does there exist an attractor of size $c$ in the dynamics
	of $\circlearrowright_\omega M$?}
\end{cproblem}

\begin{cproblem}
  \problemtitle{\textsc{One-to-one Module Attractor Problem}}
  \probleminput{A one-to-one module $M$ with $n$ nodes, a 
	function $\omega$ which defines a complete recursive wiring over $M$,
	and a number $c$ encoded in unary.}
  \problemquestion{Does there exist an attractor of size $c$ in the dynamics
	of $\circlearrowright_\omega M$?}
\end{cproblem}

\begin{cproblem}
  \problemtitle{\textsc{Acyclic Module Fixed Point Problem}}
  \probleminput{An acyclic module $M$ with $k$ inputs and $n$ nodes, and a 
	function $\omega$ which defines a complete recursive wiring over $M$.}
  \problemquestion{Does there exist a configuration $x$ such that
	$\circlearrowright_\omega M(x) = x$?}
\end{cproblem}

\begin{cproblem}
  \problemtitle{\textsc{One-to-one Module Fixed Point Problem}}
  \probleminput{A one-to-one module $M$ with $n$ nodes, and a 
	function $\omega$ which defines a complete recursive wiring over $M$.}
  \problemquestion{Does there exist a configuration $x$ such that
	$\circlearrowright_\omega M(x) = x$?}
\end{cproblem}

Those four problems are variations of the same question under different
sets of constraints. The first problem, the Acyclic Module Attractor Problem,
generalises the other three decision problems,
while the One-to-one Module Fixed Point Problem
is a specific case of the other three decision problems.
We provide our complexity analysis of those problems in a way that mirrors this
diamond-like structure.

\begin{theorem}
\label{th-fpt}
The Acyclic Module Attractor Problem can be solved in time\linebreak
$\mathcal{O}( f(k \times c) q(n) )$ for some function $f$ and $q$ a polynomial,
{\em i.e.} it is fixed parameter tractable.
\end{theorem}

\begin{proof}
We construct an algorithm which iterates all of the possible input sequences
of size $c$.
We then execute the network on each sequence and check if the outputs correspond
to the given input. This process scales polynomially with the size of the network,
but exponentially with the size of the attractor and the number of inputs.

	Formally, this algorithm checks all of the
	$|\Lambda|^{k \times c}$ possible sequences of input configurations for
	$k$ inputs and of length $c$. To check if an input configuration $J$ describes
        an attractor of size $c$ in $\circlearrowright_\omega M$, simply update
        module $M$ with an input sequence composed as the repetition of the sequence
	$J$ until the obtained sequence is at least as long as the largest delay
	in an output function of $M$. An attractor in $\circlearrowright_\omega M$
	will be obtained if for every input $\alpha$, the sequence of values of
	the node $\omega(\alpha)$ over time is identical to the sequence of values
	of the input $\alpha$. This procedure only requires in the worst case the
        evaluation of the entire network $c$ times and $k$ checks
        at each step, which is polynomial in $n \times k \times c$.

        Similarly, every possible attractor of size $c$ in $\circlearrowright_\omega M$
	has a corresponding input sequence in $M$. To see that this is true, simply
	construct an input sequence $J$ which for every input $\alpha$ defines the
	$i$-th evaluation of input $\alpha$ as the evaluation of node $\omega(\alpha)$
	in the $i$-th configuration of the attractor.

	By checking every possible input sequence for $k$ inputs and of length $c$,
        we conclude on the existence of an attractor of size $c$ in $\circlearrowright_\omega M$.
	This algorithm is of complexity $O( |\Lambda|^{k \times c}
        r(n \times k \times c) )$, for $r$ a polynomial, which implies that there exists
	$f$ a function and $q$ a polynomial such that the complexity of this
	algorithm is  $O( f(k \times c) q(n) )$. 
\end{proof}

\begin{theorem}
\label{comp-cycle}
The One-to-one Module Attractor Problem is NP-complete.
\end{theorem}
\begin{proof}
In this proof we provide a reduction from the SAT problem which for any formula
with $m$ variables, constructs a module of size $3m + 1$. The first $3m$ nodes
encode the input and the last node checks the evaluation. If at any point the
formula is evaluated at false or if the encoding is wrong, the whole network
stabilises to a fixed point. If the encoding is correct and the evaluation
positive, the configuration will shift in the network, providing an attractor
of size $3m + 1$. The existence of this attractor is proven equivalent to the
satisfiability of the formula.

  First see that this problem is in NP as,
  providing any configuration, we can verify that it
  is part of a cycle of size $c$ by updating the network $c$ times
  (each update requires to evaluate $n$ local functions)
  and making at most $c$ comparisons per step,
  for an overall polynomial time in the input size.
  
  To prove that the problem is NP-hard, we present a reduction from the 
  SAT problem. Given a formula $f$ on $m$ variables $v_1,\dots,v_m$, we
  will construct a one-to-one module on $m+e+1$ nodes
  (for some $e$ upper bounded by a constant)
  such that, when the output is wired to the input, there exists
  a cycle of size $c=m+e+1$,
  if and only if there exists a valuation
  satisfying $f$.
  
  The one-to-one module, denoted $M$, is composed of two parts.
  The first part is a {\em shifting tape}, which
  is composed of $m+e$ nodes $t_1,\dots,t_{m+e}$
  with $e$ the smallest number such that $m+e+1$ is a prime number
  (the value of $e$ is at most $2m$
  according to the Bertrand--Chebyshev theorem \cite{J-Chebychev1852},
  and one can find it in polynomial time
  thanks to the well-known algorithm from \cite{J-AKS2004}).
  For $1 < k \leq m+e$ we define the local functions $f_{t_k}(x) = t_{k - 1}$,
  and $f_{t_1}(x) = \alpha$ with $\alpha$ the only input of the network.
  For $i \in \{1,\dots,m\}$ the state of node $t_i$
  encodes the evaluation of variable $v_i$.
  
  The second part of the network is composed of a unique node denoted $q$,
  the output node to be wired to input $\alpha$, which
  has the role of either letting the shifting tape of size $m+e$ become a
  shifting tape of size $m+e+1$,
  or stopping the process and make the configuration converge
  to $0^{m+1}$. Its local function is:
  $$
    f_q(x)=\begin{array}\{{ll}.
      x_{t_{m+e}} & \text{ if nodes $t_1,\dots,t_m$ of the shifting tape encode}\\
      & \hspace*{5cm}\text{ a valuation satisfying $f$}\\
      & \text{ or a shift may encode a valuation satisfying $f$,}\\
      0 & \text{ otherwise.}
    \end{array}
  $$
  Since module $M$ is acyclic node $q$ cannot know its own state, but it knows
  the state of all other nodes. Therefore the second condition of the disjunction
  is checked as follows: node $q$ tries, for $x_q=0$ and for $x_q=1$, and for
  any $k$ from $1$ to $m+e$, whether cyclically shifting the configuration
  (considering that $q$ follows $t_{m+e}$ and preceeds $t_1$) by $k$ units
  can give a shifting tape encoding a valuation satisfying $f$ on the states of nodes
  $t_1,\dots,t_m$;
  if any combination of state for $x_q$ and shift $k$
  gives a shifting tape encoding a valuation satisfying $f$ then the condition
  ``a shift may encode a valuation satisfying $f$'' is true.

  This construction is illustrated in Figure~\ref{fig-proof}.
  It has polynomial size, as the local functions of the $c=m+e+1$ nodes
  can be expressed with propositional formulas of size polynomial in $f$ and $m+e$
  (naively for $f_q$ with a disjunction of $m+e+1$ terms,
  each containing a copy of $f$).
  
  \input{figures/proof-comp-v2}
  
  If $f$ has a satisfying valuation, then some configuration $x$ encoding this
  valuation on nodes $t_1,\dots,t_m$ of the shifting tape belongs to a cycle
  of size $c$. Indeed, in this case $x$ is cyclically shifted by one unit at each step
  along the $c=m+e+1$ nodes of $\circlearrowright M$,
  and by taking $x_q \neq x_1$ configuration $x$ cannot be a fixed point therefore
  $m+e+1$ prime implies that $F^{c'}(x) \neq x$ for all $c'<c$.
  
  If $f$ has no statisyfing valuation, then $f_q(x)=0$ for any $x$ and
  $\circlearrowright M$ has only one attractor which is a fixed point, $0^{m+e+1}$.
  
\end{proof}

\begin{theorem}
\label{comp-fix}
The Acyclic Module Fixed Point Problem is NP-complete.
\end{theorem}
\begin{proof}
This proof provides a reduction from the SAT problem. In this reduction,
the obtained module will stabilise only if a given node, which computes
a SAT formula, has constant value 1.

	First see that this problem is in NP since, given any configuration,
	verifying that it is a fixed point can be done by updating the
	whole network once, which is done in polynomial time in the size of its
	encoding.

	To see this problem as NP-hard we present a reduction from the SAT problem.
	given a formula $f$, we construct a module with one node for each
	variable in $f$. Each of these nodes are wired to themselves by the
	wiring $\omega$, forming identity local functions of the form $f_a(x) = a$.
	Then we add two other nodes to the module.
        One, named \texttt{solver}, computes $f$ from
	the states of nodes corresponding to variables.
        The second, named \texttt{oscillator},
        has local function $f_{\texttt{oscillator}}(x) =
	\neg \texttt{solver} \wedge \neg \texttt{oscillator}$.
        This is constructed
        via an input which is wired onto \texttt{oscillator} by $\omega$.

        Every node except \texttt{solver} and \texttt{oscillator} have a fixed state,
        therefore the existence of a fixed point only depends on the evaluation
        of \texttt{solver} and \texttt{oscillator}. The \texttt{solver} node has a
        fixed state after one iteration, corresponding to
	the evaluation of formula $f$ according to the states of variables nodes.
	Consequently the existence of a fixed point
        only depends on the behavior of the \texttt{oscillator} node,
        which by definition will
        oscillate as long as the evaluation of the \texttt{solver} node is $0$.
        We conclude that the existence of a fixed point in the AN obtained by wiring this
	module according to $\omega$ is equivalent to the existence of
        a positive evaluation
	of the formula $f$. This construction being polynomial in the size of the
	formula, we conclude that the problem is NP-hard.
        
\end{proof}

\begin{corollary}
The One-to-one Module Fixed Point Problem is in P.
\end{corollary}
\begin{proof}
	This is an application of Theorem~\ref{th-fpt}.
        
\end{proof}

The above stated results imply that the size of the network
is not a meaningful parameter in the difficulty of the task of finding
attractors. Thereom~\ref{th-fpt} shows that the two parameters which
apply this effect are the size of the desired attractor and the number of inputs
the network bears when seen as an acyclic module. In other terms this second
parameter is the level of interconnectivity of the network. Theorems~\ref{comp-cycle}
and~\ref{comp-fix} prove that this caracterisation is tight.
Together, these four theorems provide a new perspective on a known fact; that
cycles in ANs are crucial for complexity to arise.

\begin{cproblem}
  \problemtitle{\textsc{Acyclic Module Output Construction Problem}}
  \probleminput{A set $\{M_1, M_2, \ldots, M_\ell\}$ of acyclic modules, and
	$O$ an output function encoded in a lookup table.}
  \problemquestion{Does there exist a set of non-recursive wirings $\omega$ which can
	construct an acyclic module from $M_1, M_2, \ldots, M_\ell$ such
	that $O$ is an output function of the obtained module?}
\end{cproblem}

\begin{theorem}
The Acyclic Module Output Construction Problem is NP-complete.
\end{theorem}
\begin{proof}
We provide a reduction from the SAT problem. We ask for the construction of
an output function via the wiring of two modules with a unique constant function
\lq$0$\rq\ and \lq$1$\rq\ respectively,
and a bigger module which executes a computation from
its inputs based on the formula, such that the target output function is obtained
by non-recursive wirings if and only if the formula is satisfiable.

	This decision problem is in NP since, given the non-recursive wiring and
	the node which carries the target output function, the verification
    can be done in polynomial time. Note that the target output function is
    provided as a lookup table, and that checking the egality of two functions
    given as lookup tables can be done in a single pass, which is polynomial in
    time.

	To prove that this problem is NP-hard, take a SAT formula $f$, and construct
	the following instance of the present decision problem: the set of modules
	is $\{M_0, M_1, M_f\}$. Modules $M_0$ and $M_1$ have no input and only one node
	whose function is the constant $0$ and $1$ respectively. The
        module $M_f$ has as many inputs as the formula $f$ as variables,
        plus one denoted $\alpha$,
        and only one node which computes $f \wedge \alpha$
        using inputs corresponding to variables to compute $f$.
        The target output function $O$ is the identity (on one input) with delay $1$.

	For this instance to be positive, there has to be some wirings
        reducing the function $f \wedge \alpha$ to the identity
        (modules $M_0$ and $M_1$ have no input hence cannot produce $O$), meaning that
        the formula is satisfiable: either $\alpha$ is not wired and $f$ reduces to $1$;
        or $\alpha$ is wired (to $1$) and $f$ reduces to the identity
        on one variable, hence evaluating this last variable to $1$
        satisfies $f$.

        Conversely, if $f$ is satisfiable then wiring inputs corresponding to
        variables according to a satisfiable assignment reduces the local function 
        of module $M_f$ to $\alpha$, {\em i.e.} this node has the target output
        function $O$.
        
\end{proof}

\section{Conclusion}

%\sylvain{Il y a de la place pour dire plus peut-être...}

%Modules were initially introduced in the hope of providing a powerful formalism
%for composing and decomposing ANs. We believe that the present paper showcases the 
%utility of the formalism as a general tool to understand and predict the behavior of ANs.
%Theorem~\ref{th-limit-from-output} allow the classification of acyclic modules by
%the output function they define, while 

%Future approaches for the development of this formalism include its
%application to different update schedules.

Automata Networks are complex systems, the exhaustive study of which requires
an amount of resources exponential in the size of the network. By defining
and studying acyclic modules we propose an innovative way of approaching this
question. Theorem~\ref{th-limit-from-output} proposes
the reduction of the limit dynamic of a network to the output functions of an
acyclic module which composes it. We think that this result, alongside with
Theorem~\ref{th-one-to-one} which is a direct application of it, provides an
interesting way of categorising networks depending on their output functions.
Also presented are Theorem~\ref{th-bound} which proposes a bound on the total
number of attractors depending on the number of inputs in an acyclic module,
and the results listed in Section~\ref{section-complexity}, which state a range
of complexity results on acyclic modules.
The set of results proposed in this paper describe, in our opinion, a good
picture of the limits and possibilities that come from studying acyclic modules.

In future works, we plan to expand this formalism to more general update
schedules, and to propose a version of Theorem~\ref{th-one-to-one} which would
generalise to modules with more than one input and one output. We also plan
to apply those tools to optimise large automata networks, such as those designed
and studied in biology applications.	

\bigskip
\noindent
\textbf{Acknowledgments.}
The works of K{\'e}vin Perrot and Sylvain Sen{\'e} were funded mainly by their salaries as French State agents, % and therefore by French taxpayers' taxes,
affiliated to
Aix-Marseille Univ., Univ. de Toulon, CNRS, LIS, UMR 7020, Marseille, France (both)
and to
Univ. C{\^o}te d'Azur, CNRS, I3S, UMR 7271, Sophia Antipolis, France (KP),
and secondarily by 
ANR-18-CE40-0002 FANs project,
ECOS-Sud C16E01 project,
STIC AmSud CoDANet 19-STIC-03 (Campus France 43478PD) project.

\bibliographystyle{plain}
{\small{\bibliography{acyclic}}}

%\appendix
%\include{demonstration}

\end{document}

%% file: figures/automata-dynamics.tex
\begin{figure}[t!]
\vspace{-0.25cm}
\begin{center}
\begin{minipage}{.26\textwidth}
\begin{center}
\begin{tikzpicture} [every node/.style={scale=\ech},node distance=0.5cm,-{Latex[length=1.5mm, width=1.5mm]}]
	\node[state] (A) {a};
	\node[state] (C) [below right =of A] {c};
	\node[state] (B) [below left =of A] {b};

	\path
			(A) edge (B)
			(A) edge (C)
			(C) edge (B);

	\draw (A) to [out=110, in=70, looseness=7] (A);
	\draw (C) to [out=340, in=20, looseness=7] (C);
\end{tikzpicture}
\end{center}
\end{minipage}\qquad
\begin{minipage}{.37\textwidth}
\begin{center}
\begin{tikzpicture} [every node/.style={scale=\ech},->,node distance=1.25cm,-{Latex[length=1.5mm, width=1.5mm]}]
	\node[configuration] (000) {000};
	\node[configuration] (001) [right =of 000] {001};
	\node[configuration] (100) [above =of 000] {100};
	\node[configuration] (010) [above right =15pt of 000] {010};
	\node[configuration] (101) [right =of 100] {101};
	\node[configuration] (011) [right =of 010] {011};
	\node[configuration] (110) [above =of 010] {110};
	\node[configuration] (111) [right =of 110] {111};

	\path[-{Latex[length=2mm, width=2mm]}]
			(000) edge [bend right=42] (111)
			(100) edge (010)
			(110) edge [bend right=20] (010)
			(111) edge [bend right=35] (010)
			(011) edge [bend right=5] (100)
			(010) edge [bend right=20] (111)
			(001) edge [bend left=40] (100)
			(101) edge (010);
\end{tikzpicture}
\end{center}
\end{minipage}
\end{center}
\caption{(Left) Interaction digraph and of (right) dynamics of the network of 
  Example~\ref{example-F}.}
\label{fig-aut-dyn}
\vspace{-0.25cm}
\end{figure}
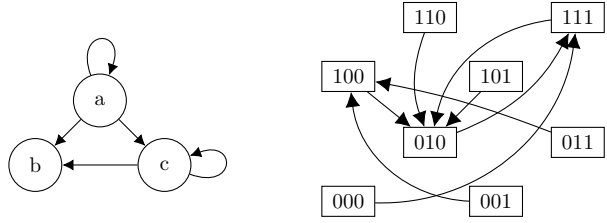

%% file: figures/wirings.tex
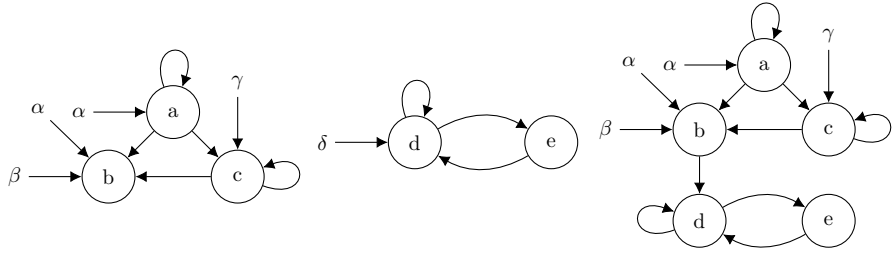
\begin{figure}[t!]
\vspace{-.25cm}
\begin{center}
\begin{minipage}{.31\textwidth}
\begin{tikzpicture} [every node/.style={scale=\ech},node distance=0.5cm,-{Latex[length=1.5mm, width=1.5mm]}]
	\node[state] (A) {a};
	\node[state] (C) [below right =of A] {c};
	\node[state] (B) [below left =of A] {b};

	\path	(A) edge (B)
			(A) edge (C)
			%(A) edge [loop above, -|] (A)
			(C) edge (B);
			%(C) edge [loop right] (C);

	\draw (A) to [out=110, in=70, looseness=7] (A);
	\draw (C) to [out=340, in=20, looseness=7] (C);

	\node [left =20pt of A] (alpha1) {$\alpha$};
	\draw (alpha1) -- (A);
	\node [above left =20pt of B] (alpha2) {$\alpha$};
	\draw (alpha2) -- (B);
	\node [left =20pt of B] (beta) {$\beta$};
	\draw (beta) -- (B);
	\node [above =20pt of C] (gamma) {$\gamma$};
	\draw (gamma) -- (C);
\end{tikzpicture}
\end{minipage}\quad
\begin{minipage}{.28\textwidth}
\begin{tikzpicture} [every node/.style={scale=\ech},node distance=0.5cm,-{Latex[length=1.5mm, width=1.5mm]}]
	\node[state] (D) {d};
	\node[state] (E) [right =1.1cm of D] {e};

	\path	(D) edge [bend left] (E)
			(E) edge [bend left] (D);

	\draw (D) to [out=110, in=70, looseness=7] (D);

	\node [left =20pt of D] (delta) {$\delta$};
	\draw (delta) -- (D);
\end{tikzpicture}
\end{minipage}\quad
\begin{minipage}{.31\textwidth}
\begin{tikzpicture} [every node/.style={scale=\ech},node distance=0.5cm,-{Latex[length=1.5mm, width=1.5mm]}]
	\node[state] (A) {a};
	\node[state] (C) [below right =of A] {c};
	\node[state] (B) [below left =of A] {b};

	\path	(A) edge (B)
			(A) edge (C)
			%(A) edge [loop above, -|] (A)
			(C) edge (B);
			%(C) edge [loop right] (C);

	\draw (A) to [out=110, in=70, looseness=7] (A);
	\draw (C) to [out=340, in=20, looseness=7] (C);

	\node [left =20pt of A] (alpha1) {$\alpha$};
	\draw (alpha1) -- (A);
	\node [above left =20pt of B] (alpha2) {$\alpha$};
	\draw (alpha2) -- (B);
	\node [left =20pt of B] (beta) {$\beta$};
	\draw (beta) -- (B);
	\node [above =20pt of C] (gamma) {$\gamma$};
	\draw (gamma) -- (C);

	\node[state] (D) [below =of B] {d};
	\node[state] (E) [below =of C] {e};

	\path	(D) edge [bend left] (E)
			(E) edge [bend left] (D);

	\draw (D) to [out=200, in=160, looseness=7] (D);

	\path (B) edge (D);
\end{tikzpicture}
\end{minipage}
\end{center}
\caption{Illustration of the wiring of Example~\ref{example-omega}. Interaction
digraphs of the modules (left) $M$, (center) $M'$ and (right) $M''$.}
\label{fig-wiring}
\vspace{-0.25cm}
\end{figure}

%% file: figures/acyclic-module.tex
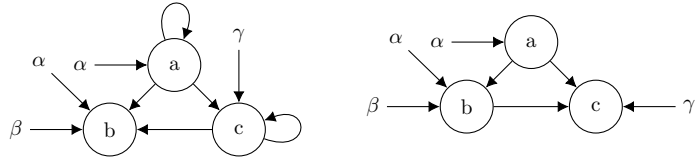
\begin{figure}[t!]
\vspace{-.25cm}
\begin{center}
\begin{minipage}{.4\textwidth}
\begin{center}
\begin{tikzpicture} [every node/.style={scale=\ech},node distance=0.5cm,-{Latex[length=1.5mm, width=1.5mm]}]
	\node[state] (A) {a};
	\node[state] (C) [below right =of A] {c};
	\node[state] (B) [below left =of A] {b};

	\path	(A) edge (B)
			(A) edge (C)
			%(A) edge [loop above, -|] (A)
			(C) edge (B);
			%(C) edge [loop right] (C);

	\draw (A) to [out=110, in=70, looseness=7] (A);
	\draw (C) to [out=340, in=20, looseness=7] (C);

	\node [left =20pt of A] (alpha1) {$\alpha$};
	\draw (alpha1) -- (A);
	\node [above left =20pt of B] (alpha2) {$\alpha$};
	\draw (alpha2) -- (B);
	\node [left =20pt of B] (beta) {$\beta$};
	\draw (beta) -- (B);
	\node [above =20pt of C] (gamma) {$\gamma$};
	\draw (gamma) -- (C);

\end{tikzpicture}
\end{center}
\end{minipage}
\begin{minipage}{.4\textwidth}
\begin{center}
\begin{tikzpicture} [every node/.style={scale=\ech},node distance=0.5cm,-{Latex[length=1.5mm, width=1.5mm]}]
	\node[state] (A) {a};
	\node[state] (C) [below right =of A] {c};
	\node[state] (B) [below left =of A] {b};

	\path	(A) edge (B)
			(A) edge (C)
			(B) edge (C);

	\node [left =20pt of A] (alpha1) {$\alpha$};
	\draw (alpha1) -- (A);
	\node [above left =20pt of B] (alpha2) {$\alpha$};
	\draw (alpha2) -- (B);
	\node [left =20pt of B] (beta) {$\beta$};
	\draw (beta) -- (B);
	\node [right =20pt of C] (gamma) {$\gamma$};
	\draw (gamma) -- (C);

\end{tikzpicture}
\end{center}
\end{minipage}
\end{center}
\caption{Interaction digraph of (left) the module of Example~\ref{example-M}, (right) 
the acyclic module of Example~\ref{example-acyM}.}
\label{fig-acyclic}
\vspace{-0.25cm}
\end{figure}

%% file: figures/one-to-one.tex
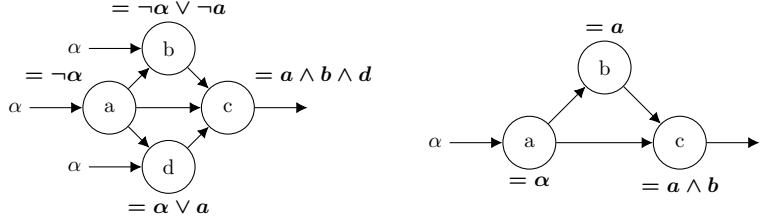
\begin{figure}[t!]
\vspace{-0.5cm}
\begin{center}
\begin{minipage}{.4\textwidth}
\begin{tikzpicture} [every node/.style={scale=\ech},node distance=0.4cm,-{Latex[length=1.5mm, width=1.5mm]}]
	\node[state] (A) {a};
	\node[state] (B) [above right=of A] {b};
	\node[state] (C) [below right=of B] {c};
	\node[state] (D) [below left =of C] {d};

	\path	(A) edge (B)
			(A) edge (C)
			(A) edge (D)
			(B) edge (C)
			(D) edge (C);

	\node [left =20pt of A] (alpha1) {$\alpha$};
	\draw (alpha1) -- (A);
	\node [left =20pt of B] (alpha2) {$\alpha$};
	\draw (alpha2) -- (B);
	\node [left =20pt of D] (alpha3) {$\alpha$};
	\draw (alpha3) -- (D);
	\node [right =20pt of C] (output) {};
	\draw (C) -- (output);

	\node [above left=0pt of A] (AF) {$\boldsymbol{=\neg \alpha}$};
	\node [above=0pt of B] (BF) {$\boldsymbol{=\neg \alpha \vee \neg a}$};
	\node [above right=0pt of C] (CF) {$\boldsymbol{=a \wedge b \wedge d}$};
	\node [below=0pt of D] (DF) {$\boldsymbol{=\alpha \vee a}$};
\end{tikzpicture}
\end{minipage}
\hspace{0.5cm}
\begin{minipage}{.4\textwidth}
\begin{tikzpicture} [every node/.style={scale=\ech},node distance=0.7cm,-{Latex[length=1.5mm, width=1.5mm]}]
	\node[state] (A) {a};
	\node[state] (B) [above right =of A] {b};
	\node[state] (C) [below right =of B] {c};

	\path	(A) edge (B)
			(A) edge (C)
			(B) edge (C);

	\node [left =20pt of A] (alpha) {$\alpha$};
	\draw (alpha) -- (A);
	\node [right =20pt of C] (output) {};
	\draw (C) -- (output);

	\node [below=0pt of A] (AF) {$\boldsymbol{=\alpha}$};
	\node [above=0pt of B] (BF) {$\boldsymbol{=a}$};
	\node [below =0pt of C] (CF) {$\boldsymbol{=a \wedge b}$};
\end{tikzpicture}
\end{minipage}
\end{center}
\caption{Illustration of Theorem~\ref{th-one-to-one}. Both modules consider
$c$ as their output node, and display the same output function
$O = \alpha_2 \wedge \alpha_3$. The module on the right is optimal, as $3$ is
the delay of its output function.}
\label{fig-one-to-one}
\end{figure}

%% file: figures/proof-comp-v2.tex
\tikzset{small state/.style = {state, minimum size = 0.8cm},
}

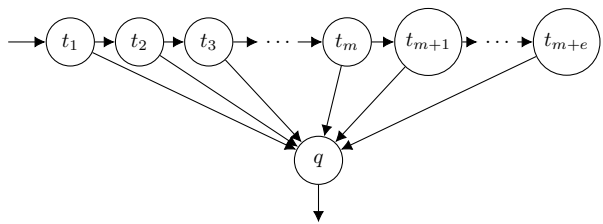
\begin{figure}[t!]
\begin{center}
\begin{tikzpicture} [every node/.style={scale=\ech},-{Latex[length=1.5mm, width=1.5mm]}, node distance =1.15cm]

	\node[small state] (t1) {$t_1$};
	\node[small state] (t2) [right of=t1] {$t_2$};
	\node[small state] (t3) [right of=t2] {$t_3$};
	\node (dots) [right of=t3] {$\ldots$};
	\node[small state] (tm) [right of=dots] {$t_m$};
        \node[small state] (tm1) [right =0.3cm of tm] {$t_{m+1}$};
	\node (dots2) [right of=tm1] {$\ldots$};
        \node[small state] (tme) [right of=dots2] {$t_{m+e}$};

        \path (t1) edge (t2)(t2) edge (t3)(t3) edge (dots)(dots) edge (tm)(tm) edge (tm1)(tm1) edge (dots2)(dots2) edge (tme);

	\node (center) at ($(t1)!0.5!(tme)$) {};
	\node[small state] (q) [below = of center] {$q$};

        \path (t1) edge (q);
        \path (t2) edge (q);
        \path (t3) edge (q);
        \path (tm) edge (q);
        \path (tm1) edge (q);
        \path (tme) edge (q);

	\node[left =0.5cm of t1] (input) {};
	\node[below =0.5cm of q] (output) {};

	\path (input) edge (t1) (q) edge (output);
\end{tikzpicture}
\end{center}
\caption{Module $M$ in the proof of Theorem~\ref{comp-cycle}. 
If $f$ has a satisfying valuation then node $q$ can let the shifting tape of size $m+e$
become a rotating tape of size $m+e+1$, otherwise $f_q$ evaluates to $0$ and any
configuration converges to the fixed point $0^{m+e+1}$.}
\label{fig-proof}
\end{figure}